\documentclass[runningheads]{llncs}
\usepackage[T1]{fontenc}
\usepackage{graphicx}
\usepackage{color}
\usepackage[nolist,printonlyused]{acronym}
\usepackage{amsmath, amssymb, amsfonts}

\allowdisplaybreaks
\usepackage{amsthm}
\usepackage{algpseudocode}
\usepackage{balance}
\usepackage{bm}
\usepackage[english]{datetime2}
\usepackage{glossaries}
\usepackage{multicol}
\usepackage{nicefrac}
\usepackage{textcomp}
\usepackage{tikz}
\usepackage{xcolor}
\usepackage{textcomp}

\usepackage[ruled,vlined]{algorithm2e}

\newcommand{\algcmt}[1]{\tcp*[r]{#1}}   
\SetAlgoInsideSkip{0pt}

\usepackage{sidecap}
\sidecaptionvpos{figure}{t}

\DeclareMathOperator*{\argmin}{\arg\!\min}

\usepackage[parse-numbers=false]{siunitx}

\usepackage[font=footnotesize]{subcaption}

\usepackage{hyperref}
\usepackage{cleveref}

\newtheorem{assumption}{Assumption}

\crefname{figure}{Fig.}{Figs.}
\crefname{equation}{}{}

\crefname{table}{Table}{Tables}
\crefname{algorithm}{Alg.}{Algorithms}
\crefname{section}{Section}{Sections}
\crefname{remark}{Remark}{Remarks}
\def\BibTeX{{\rm B\kern-.05em{\sc i\kern-.025em b}\kern-.08em
    T\kern-.1667em\lower.7ex\hbox{E}\kern-.125emX}}

\title{\LARGE \bf 
Efficiently Solving Mixed-Hierarchy Games with Quasi-Policy Approximations 
}

\titlerunning{Mixed-Hierarchy Games with Quasi-Policy Approximations}
\authorrunning{H.~I.~Khan et al.}

\author{
Hamzah I. Khan\inst{1*}
\orcidID{0000-0001-5481-6388}
\and
Dong Ho Lee\inst{1}
\orcidID{0000-0002-9045-2574}
\and
Jingqi Li\inst{1} 
\orcidID{0000-0002-3731-3807}
\and
Tianyu Qiu\inst{1}
\orcidID{0009-0008-5189-3722}
\and
Christian Ellis\inst{1,2} 
\orcidID{0000-0002-0544-8627}
\and
Jesse Milzman\inst{2} 
\orcidID{0000-0003-4937-8912}
\and
Wesley Suttle\inst{2} 
\orcidID{0000-0003-1234-7151} 
\and
David Fridovich-Keil\inst{1}
\orcidID{0000-0002-5866-6441}
}

\institute{
University of Texas at Austin, Austin, TX 78712, USA\\
\email{Corresponding Author\inst{*}: hamzah@utexas.edu}
\and
U.S. Army Research Laboratory, Adelphi, MD 20783, USA
}

\crefname{figure}{Fig.}{Figs.}
\crefname{equation}{}{}
\usepackage{comment}
\crefname{table}{Table}{Tables}
\crefname{algorithm}{Alg.}{Algorithms}
\crefname{section}{Section}{Sections}
\crefname{remark}{Remark}{Remarks}
\def\BibTeX{{\rm B\kern-.05em{\sc i\kern-.025em b}\kern-.08em
    T\kern-.1667em\lower.7ex\hbox{E}\kern-.125emX}}

\newcommand\extrafootertext[1]{%
    \bgroup
    \renewcommand\thefootnote{\fnsymbol{footnote}}%
    \renewcommand\thempfootnote{\fnsymbol{mpfootnote}}%
    \footnotetext[0]{#1}%
    \egroup
}

\newcommand\MethodName{Mixed Hierarchy Game}

\newcommand\edits[1]{\textcolor{blue}{#1}}
\newcommand\newedit[1]{\textcolor{purple}{#1}}

\newcommand\cmt[1]{\textcolor{purple}{#1}}
\newcommand\todo[1]{\textcolor{red}{TODO: #1}}

\newcommand\david[1]{\textcolor{orange}{[DFK: #1]}}
\newcommand\hmzh[1]{\textcolor{purple}{[HK: #1]}}
\newcommand\jingqi[1]{\textcolor{blue}{[JL: #1]}}

\renewcommand\todo[1]{T_}
\renewcommand\hmzh[1]{H_}
\renewcommand\david[1]{D_}
\renewcommand\jingqi[1]{J_}
\renewcommand\newedit[1]{N_}
\renewcommand{\cmt}[1]{#1}
\renewcommand\edits[1]{#1}

\newcommand{\T}{\top}

\newcommand{\decVarSymbol}{z}
\newcommand\decVar[1]{\mathbf{\decVarSymbol{}}^{#1}}
\newcommand\allDecVars{\decVar{}}

\newcommand\numstates[1]{n_x^{#1}}
\newcommand\numctrls[1]{n_u^{#1}}
\newcommand\numDecVars[1]{n_z^{#1}}

\newcommand\numtotalvars[1]{n_s^{#1}}

\newcommand\initstate[1]{\mathbf{x}_{\textrm{init}}^{#1}}
\newcommand\state[2]{\mathbf{x}^{#1}_{#2}}

\newcommand\ctrl[2]{\mathbf{u}^{#1}_{#2}}

\newcommand\totalvar[1]{\bm{s}^{#1}}
\newcommand\totalvariter[2]{\bm{s}^{#1}_{#2}}

\newcommand\outvar[1]{\mathbf{w}^{#1}}

\newcommand\equalityconstraint[1]{g^{#1}}

\newcommand{\constraintDualSym}{\lambda{}}
\newcommand{\constraintDual}[1]{\constraintDualSym{}^{#1}}

\newcommand{\policyDualSym}{\mu{}}
\newcommand{\policyDual}[2]{\policyDualSym{}^{#1,#2}}
\newcommand{\allFollowerPolicyDuals}[1]{\bm{\policyDualSym{}}^{#1}}

\newcommand\lagrangian[1]{\mathcal{L}^{#1}}

\newcommand\numpolicyduals[1]{n_{\mu}^{#1}}
\newcommand\numequalitylagrange[1]{n_g^{#1}}

\newcommand{\numplayers}{N}
\newcommand{\numrobots}{N}

\newcommand\idxi{i}
\newcommand\idxj{j}
\newcommand\idxk{k}
\newcommand\iter{k}

\newcommand\horizon{T}
\newcommand\dt{\Delta t}

\newcommand\dynamics[1]{f^{#1}}
\newcommand\costfn[1]{J^{#1}}

\newcommand\allpolicies[1]{\mathrm{\Phi}^{#1}}
\newcommand\policyfn[2]{\mathrm{\phi}^{#2}}

\newcommand\approxkktConds[1]{\pi^{#1}}

\newcommand\kktConds[1]{\pi^{#1}}

\newcommand\kktGuess[1]{b^{\idxi{}}}

\newcommand{\stepsize}[1]{\alpha_{#1}}

\newcommand\graph{\mathcal{G}}
\newcommand\edgeset{\mathcal{E}}
\newcommand\Nset{[\numplayers{}]}

\newcommand\tree[1]{\mathcal{T}_{#1}}

\newcommand\leader{\mathrm{L}}
\newcommand\follower{\mathrm{F}}
\newcommand\nash{\mathrm{N}}
\newcommand\nashset[1]{\boldsymbol{\nash}(#1)}
\newcommand\leaderset[1]{\leader(#1)}
\newcommand\allleaderset[1]{\boldsymbol{\leader}(#1)}
\newcommand\followerset[1]{\follower(#1)}
\newcommand\allfollowerset[1]{\boldsymbol{\follower{}}(#1)}

\newcommand\optimalresponseset[2]{\mathcal{R}^{#1*}_{\follower{}}(#2)}

\newcommand\arbitraryint{d}
\newcommand\arbitraryset{\mathcal{S}}
\newcommand\arbitrarysetelems[1]{e^{#1}}
\newcommand\arbitraryindexset{E}

\newcommand{\threshold}{\tau}
\newcommand{\maxiters}{K_{\max}}

\begin{document}

\maketitle

\begin{abstract}
    Multi-robot coordination often exhibits hierarchical structure, with some robots' decisions depending on the planned behaviors of others. While game theory provides a principled framework for such interactions, existing solvers struggle to handle mixed information structures that combine simultaneous (Nash) and hierarchical (Stackelberg) decision-making. 
We study $\numplayers{}$-robot forest-structured mixed-hierarchy games, in which each robot acts as a Stackelberg leader over its subtree while robots in different branches interact via Nash equilibria. 
We derive the Karush–Kuhn–Tucker (KKT) first-order optimality conditions for this class of games and show that they involve increasingly high-order derivatives of robots’ best-response policies as the hierarchy depth grows, rendering a direct solution intractable. 
To overcome this challenge, we introduce a quasi-policy approximation that removes higher-order policy derivatives and develop an inexact Newton method for efficiently solving the resulting approximated KKT systems. 
We prove local exponential convergence of the proposed algorithm for games with nonquadratic objectives and nonlinear constraints. 
The approach is implemented in a highly optimized Julia library (\texttt{MixedHierarchyGames.jl}\extrafootertext{Code available: \url{https://github.com/CLeARoboticsLab/MixedHierarchyGames.jl}.}) and evaluated \cmt{in hardware and} simulated multi-agent experiments,
demonstrating real-time convergence for complex mixed-hierarchy information structures.
\extrafootertext{This work was sponsored by the Army Research Laboratory under Cooperative Agreement Number W911NF-25-2-0021, and by the National Science Foundation under Grants 2211548 and 2336840.}
\end{abstract}

\section{Introduction}
\label{sec:intro}


Strategic multi-robot decision-making often involves hierarchical information structures that define leader and follower roles among different robots, each of which makes a constrained decision to optimize a distinct objective with access to different information.
Mathematical games, particularly the Nash and Stackelberg equilibrium concepts, provide a powerful framework to model such strategic interactions.
For example, consider a convoy of three vehicles driving in a single lane.
In such a scenario, we may model each vehicle as a leader of each subsequent vehicle in the lane, resulting in a leadership ``chain'' that prior works model within the formalism of Stackelberg games \cite{basar1981trilevel,shafiei2024trilevel,li2024computation}.

However, many scenarios require modeling a  
mix of Nash and Stackelberg relationships between groups of robots \cite{li2025iesep,mohapatra2025feedback}.
\cmt{Beyond convoy merging, mixed-hierarchy structures arise in multi-robot patrolling, social navigation, and human-robot teams, where designated leaders influence subgroups whose members interact as peers.}
Consider a fourth car seeking to merge into the convoy's lane between the first and second convoy vehicles, as in \cref{fig:merging-experiments}.
Merging into this position in the convoy requires the fourth car to coordinate with both other vehicles.
First, the merging vehicle must respond to signaled behavior from the first car, such as slowing down for other vehicles or speeding up to make space.
The same type of signaling does not necessarily exist between the merging robot and the second car in the convoy, making the Nash equilibrium model of simultaneous play a natural model for their interaction.

\begin{figure*}[!t]
\centering
\includegraphics[width=0.8\textwidth]{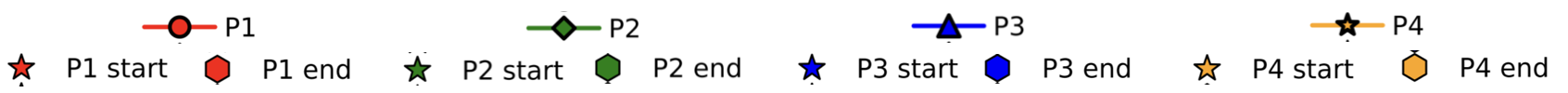}
\begin{subfigure}{0.24\textwidth}
    \centering
    \includegraphics[width=\textwidth]{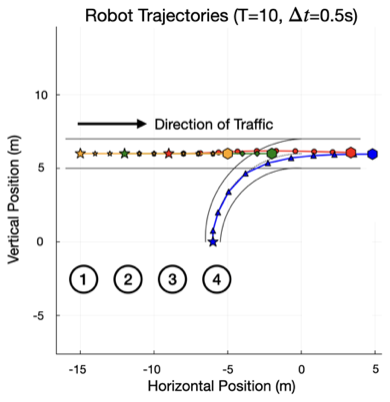}
    \caption{Nash hierarchy structure.}
    \label{subfig:merging-nash}
\end{subfigure}
\hfill
\begin{subfigure}{0.24\textwidth}
    \centering
    \includegraphics[width=\textwidth]{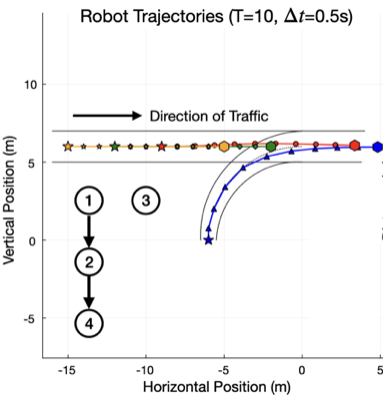}
    \caption{Mixed hierarchy structure A.}
    \label{subfig:merging-mixed-a}
\end{subfigure}
\hfill
\begin{subfigure}{0.24\textwidth}
    \centering
    \includegraphics[width=\textwidth]{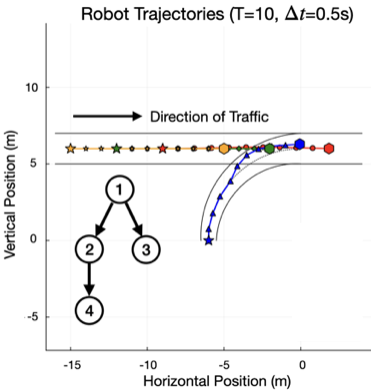}
    \caption{Mixed hierarchy structure B.}
    \label{subfig:merging-mixed-b}
\end{subfigure}
\hfill
\begin{subfigure}{0.24\textwidth}
    \centering
    \includegraphics[width=\textwidth]{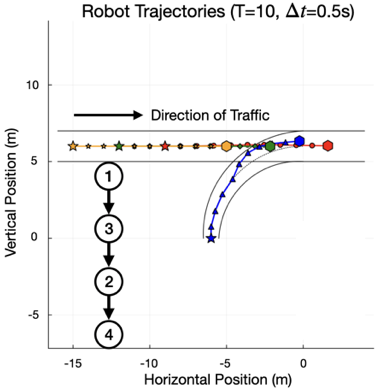}
    \caption{Stackelberg chain structure.}
    \label{subfig:merging-stack-chain}
\end{subfigure}

\vspace{-6pt}
\caption{
\label{fig:merging-experiments}
We experiment on four hierarchy structures on a merging scenario and report positions over time.
\cmt{In \cref{subfig:merging-nash}, vehicles play a pure Nash game.
In \cref{subfig:merging-mixed-a}, vehicle 1 leads vehicles 2 and 4, and vehicle 2 leads vehicle 4.
In \cref{subfig:merging-mixed-b}, vehicle 1 leads every other agent and vehicle 2 leads vehicle 4.
Vehicles 3 and (2,4) are negotiating the merge simultaneously.
In \cref{subfig:merging-stack-chain}, vehicles play a Stackelberg chain with hierarchy $1\to3\to2\to4$.
In the absence of leadership from vehicle 1 to vehicle 3, vehicle 3 merges first (\cref{subfig:merging-nash,subfig:merging-mixed-a}). When that leadership exists, vehicle 3 merges into the convoy after vehicle 1 (\cref{subfig:merging-mixed-b,subfig:merging-stack-chain}).}
}
\vspace{-18pt}
\end{figure*}

While some approaches model this type of mixed Nash-Stackelberg 
information structure \cite{laine2024mathematical,li2025iesep,mohapatra2025feedback}, they typically solve problems with linear constraints and quadratic costs, limiting their applicability to real-world robotic systems that feature nonlinear constraints (often arising from nonlinear motion dynamics) and nonquadratic objectives \cmt{(i.e., not restricted to quadratic polynomials in the decision variables)} that encode complex tasks.
To address this gap, we develop a more general solver that efficiently handles nonquadratic objectives and nonlinear equality constraints for a subclass of mixed-hierarchy games with leadership structures that are forests. This subclass of games includes Nash and Stackelberg games as special cases.

A key difficulty arises, however, in expressing the first-order optimality conditions for these mixed-hierarchy games: these conditions depend upon high-order derivatives of robots' best-response maps, or policies. 
To maintain tractability, we develop a ``quasi-policy'' approximation which extends recent work in the feedback Nash \cite{laine2023computation} and Stackelberg \cite{li2024computation} settings, wherein each robot's first-order optimality conditions are derived by ignoring higher-order derivatives of downstream followers' policies.
This approximation enables us to solve mixed-hierarchy games efficiently, with real-time computation demonstrating the value of the quasi-policy approximation in enabling real-world motion planning and robot coordination applications.
Ultimately, we claim three contributions. 
\begin{enumerate}
    \item We propose a principled approach for constructing approximate optimality conditions for a class of nonlinear, equality-constrained mixed-hierarchy games with forest-structured information patterns, including a novel method for embedding followers' best responses within leaders' optimality conditions as \emph{locally first-order} policy approximations (which we refer to as a quasi-policy approximation, following \cite{laine2023computation,li2024computation}).
    \item We introduce
    an efficient, flexible, and scalable algorithm for settings in which robots have arbitrary nonlinear, differentiable costs and \cmt{(equality)} constraints,
    that uses these optimality conditions to approximate local equilibria for arbitrary forest-structured mixed-hierarchy games. We provide a theoretical proof and an empirical study for the local exponential convergence of our algorithm to an approximate mixed-hierarchy equilibrium.
    \item 
    \edits{We provide an open-source, optimized Julia library for solving these problems, which supports simulations of the merging scenario in \cref{fig:merging-experiments} as well as in \cmt{hardware with}
    robots playing a target-guarding game in \cref{fig:front-figure}.} We demonstrate real-time
    motion planning  
    under complex hierarchical leadership structures.
\end{enumerate}

\section{Related Work}
\label{sec:related-work}
\subsubsection{Nash and Stackelberg Games.}
Classical Nash and Stackelberg equilibria capture mutual best responses \cite{nash2024noncooperative} and leader-follower interactions via sequential choices \cite{stackelberg1952theory}, respectively. 
Although closed-form expressions for the solutions to these problems are available in highly-structured cases---e.g., when agents' objectives are quadratic and only linear equality constraints are present \cite{bacsar1998dynamic}---more general settings typically have no clear analytic solution and must be approached computationally, from the perspective of mathematical programming and complementarity-based formulations \cite{fortuny1981representation,luo1996mathematical}. These formulations enable \emph{local} equilibrium computation for Nash and Stackelberg games with complex objectives and constraints \cite{fridovich2020efficient,laine2023computation,li2024computation}. 
However, they inherently assume simple information structures such as simultaneous interaction (Nash) or a direct leader-follower hierarchy (Stackelberg).
Extensions exist for multi-leader or multi-follower settings \cite{sherali1983stackelberg,kulkarni2014shared,mohapatra2025feedback}, but do not address nonlinear games under mixed Nash-Stackelberg leadership structures. Our work targets this gap.

\subsubsection{Computing Equilibria under Hierarchical Information Structures.}
Optimization problems which encode a follower's decision problem as a constraint upon the leader's decision are naturally bilevel in nature \cite{bracken1973mathematical,shi1981general}.
Moreover, as longer leader-follower chains arise, these hierarchical, multilevel problems' \cite{basar1981trilevel,candler1977multi} computational complexity grows rapidly  \cite{jeroslow1985polynomial,blair1992computational}.
Recent work explores monotone operator formulations \cite{shafiei2024trilevel}, gradient-based methods \cite{sato2021gradient}, and parametric techniques \cite{faisca2009multi} for these problems, but these approaches largely focus on low-dimensional settings. 
More expressive frameworks include quadratic constrained games under bilevel tree-structured hierarchies \cite{carvalho2024nash} and the mathematical programming network (MPN) framework \cite{laine2024mathematical}, which encodes arbitrary mixtures of Nash and Stackelberg relations via symbolic graphs. 
While the MPN framework fundamentally describes the same type of mixed Nash-Stackelberg hierarchies we consider in this work, existing MPN solvers are restricted to quadratic objectives and linear inequality constraints, and face scalability challenges due to polytope-based feasible-set representations and iterative active-set updates that must be repeatedly refined during a recursive graph traversal.
In contrast, our work solves the nonquadratic MPN problem in the nonlinear equality-constrained setting and supports tree-structured hierarchies of arbitrary depth.

\subsubsection{Stochastic Game Perspectives on Information Structure. }
A large body of work models information asymmetry and partial observability in sequential decision-making using stochastic and Markov games \cite{shapley1953stochastic,littman1994markov,filar2012competitive}. Within this framework, mixed Nash-Stackelberg and incomplete-information settings have been studied via regret-based and value-based learning \cite{zinkevich2007regret,chang2015value,farina2021model,kozuno2021learning,altabaa2024role}, as well as belief-state dynamic programming for partially observable stochastic games \cite{emery2004approximate,hansen2004dynamic,liu2022sample,becker2024bridging}.
Another line of work studies correlated equilibria under various information structures \cite{aumann1987correlated,brandenburger1992correlated,castiglioni2019leader}.
These approaches primarily focus on stochastic dynamics and discrete state-action spaces.
In contrast, we consider deterministic dynamic games with continuous actions under mixed Nash and Stackelberg information structures, and efficiently compute equilibria using differentiable mathematical programming rather than belief-space planning or model-free learning.

\section{Problem Formulation}
\label{sec:formulation}

\noindent\textbf{Set Notation.} For a positive integer $\arbitraryint{}$, let $[\arbitraryint{}] \!=\! \{1, 2, \ldots, \arbitraryint{}\}$. 
Let $\arbitraryset{}$ be an arbitrary finite set with cardinality denoted $|\arbitraryset{}|$.
For a set of indices $\arbitraryindexset{} \subseteq [|\arbitraryset{}|]$,
we define set-based element lookup $\arbitraryset{}^{\arbitraryindexset{}} = \{ \arbitrarysetelems{\idxi{}} ~|~ \idxi \in \arbitraryindexset{}, \arbitrarysetelems{\idxi{}} \in \arbitraryset{} \}$ and exclusion by $\arbitraryset{}^{-\arbitraryindexset{}} = \arbitraryset{} \setminus \arbitraryset{}^{\arbitraryindexset{}}$. 

\noindent\textbf{Decision Variables.} 
We consider $\numrobots{}$ interacting robots.
Each robot $\idxi{} \in [\numrobots{}]$ has trajectory $\decVar{\idxi{}} \in \mathbb{R}^{\numDecVars{\idxi{}}}$.
We denote the joint trajectories by $\allDecVars{} = ( \decVar{1}, \cdots, \decVar{\numrobots{}} )^\T \in \mathbb{R}^{\numDecVars{}}$, where $\numDecVars{} = \sum_{\idxi{}=1}^{\numrobots{}}  \numDecVars{\idxi{}}$.

\subsection{Mixed Hierarchy Games}
In an $\numrobots{}$-robot \emph{mixed hierarchy game}, each robot $\idxi{} \in \Nset{}$ minimizes an
objective function of the strategies $\decVar{}$ of all robots,
\begin{equation}
    \label{eq:objective}
    \costfn{\idxi{}}: \mathbb{R}^{\numDecVars{}} \to \mathbb{R},
\end{equation}
%
subject to a 
set of equality constraints 
which restrict its strategy, e.g. by enforcing dynamic feasibility:
\begin{equation}
    \label{eq:equality-constraints}
    0 = \equalityconstraint{\idxi{}}(\decVar{\idxi{}}) \in \mathbb{R}^{\numequalitylagrange{i}}.
\end{equation}
\begin{assumption}
We require \cref{eq:objective,eq:equality-constraints} to satisfy higher-order differentiability, i.e. $\costfn{\idxi}, \equalityconstraint{\idxi{}} \in C^\infty$.
\end{assumption}

Finally, we model the information structure (i.e., the set of hierarchical leader-follower relationships between robots)
as a directed acyclic graph $\graph{} = (\Nset{}, \edgeset{})$, i.e., \emph{without any directed cycles}. 
In $\graph{}$, each robot $\idxi{} \in \Nset{}$ corresponds to a node and $(\idxi{}, \idxj{}) \in \edgeset{}$ defines a direct \emph{leader-follower relationship} between leader $\idxi{}$ and follower $\idxj{}$; if $(\idxi{}, \idxj{})$ are connected by a longer path, we say that they have an \emph{indirect} leader-follower relationship.
Thus, a mixed-hierarchy game consists of graph $\graph$ and the set of all robots' objectives and equality constraints.


\subsection{Notation for Defining Mixed-Hierarchy Equilibria}
\label{ssec:notation-mixed-equilibria}
We will restrict our attention to graphs in which each node has at most one parent, i.e., each robot has at most one direct leader.
To this end, we introduce additional notation and an assumption on the topology of graph $\graph{{}}$. 
First, we define the sets of \emph{direct leaders} of a robot $\idxj{}$ and \emph{direct followers} of a robot $\idxi{}$ as 
\begin{equation}
    \label{eq:leader-follower-sets}
    \leaderset{\idxj{}} = \{ \idxi{} \mid (\idxi{}, \idxj{}) \in \graph{} \} 
    ~~~ \text{and} ~~~
    \followerset{\idxi{}} = \{ \idxj{} \mid (\idxi{}, \idxj{}) \in \graph{} \}.
\end{equation}
Next, we define the direct leaders of a set of robot indices $\arbitraryset{}$ as the union of the direct leaders of each robot, $\leaderset{\arbitraryset{}} = \bigcup_{\arbitrarysetelems{} \in \arbitraryset{}}\leaderset{\arbitrarysetelems{}}$.
We define the direct followers of $\arbitraryset{}$ similarly, $\followerset{\arbitraryset{}} = \bigcup_{\arbitrarysetelems{} \in \arbitraryset{}}\followerset{\arbitrarysetelems{}}$.
Finally, we recursively define their transitive closures, i.e., the sets of all leaders and followers (whether direct or indirect) of robot $\idxi{}$ to be
\begin{equation}
    \allleaderset{\idxi{}} = \leaderset{\idxi{}} \cup \leaderset{\leaderset{\idxi{}}} \cup \cdots ~\text{and}~
    \allfollowerset{\idxi{}} = \follower(\idxi{}) \cup \followerset{\followerset{\idxi{}}} \cup \cdots.
\end{equation}
Since $\graph{}$ is finite and acyclic, then these sets are finite too.
Finally, we define a set corresponding to the indices of robots with no other relationship to robot $\idxi{}$,
\begin{equation}
\nashset{\idxi{}} = \Nset{} \setminus \allleaderset{\idxi{}}{} \setminus \allfollowerset{\idxi{}}{} \setminus \{\idxi{}\}.
\end{equation}
These correspond to robots that have Nash relationships with robot $\idxi{}$, as we discuss in further detail below.

To make the problem tractable and ensure a clear leadership hierarchy, we restrict the topology of the information structure to one made of a union of a finite number of disconnected trees, i.e., a forest.
Such a topology remains versatile and is able to model a wide variety of hierarchical structures of interest, including pure Nash (for which $\edgeset{} = \{\}$) and Stackelberg chains (for which $\edgeset{} = \{ (\idxi{}, \idxi{}+1) ~|~ \idxi{} \in [\numrobots{}-1] \}$).
\begin{assumption}[The Hierarchical Information Structure Graph $\graph{}$ is a Forest]
\label{ass:forest-info-structure}
We assume that every robot $\idxi{} \in \Nset{}$ has either zero or one direct leader, i.e. $|\leaderset{\idxi{}}| \in \{0, 1\}$ (though it can have multiple indirect leaders). 
This condition ensures that every connected component of $\graph{}$ is a tree (i.e. $\graph{}$ is a forest) and that every node $\idxi{} \in \Nset{}$ is the root of a subtree, and therefore a leader to every robot in that subtree $\tree{\idxi{}} = (\{\idxi{}\} \cup \allfollowerset{\idxi{}}, \{ (\idxj{}, \idxk{}) \in \edgeset{} ~|~ \idxj{}, \idxk{} \in \{\idxi{}\} \cup \allfollowerset{\idxi{}} \})$.
We provide examples of (in)valid hierarchy structures in \cref{fig:graph-definitions,fig:invalid-graphs}~.
\end{assumption}

\begin{figure*}[!t]
\centering
\hfill
\begin{minipage}{0.29\textwidth}
    \centering
\begin{tabular}{c|cccc}
       $\idxi{}$ & \textbf{1}           & \textbf{2}         & \textbf{3}         & \textbf{4}        \\
\hline
$\leaderset{\idxi{}}$       
        &--         & 1     & 1                  & 2                 \\
$\allleaderset{\idxi{}}$   
        &--         & 1                  & 1                  & 1,2               \\
$\followerset{\idxi{}}$        
        & 2,3                  & 4      & --      & --     \\
$\allfollowerset{\idxi{}}$  
        & 2,3,4                & 4                  & --      & --     \\
$\nashset{\idxi{}}$        
        & --        & 3                  & 2,4          & 3$$           \\
\end{tabular}
\end{minipage}
\hfill
\begin{minipage}{0.29\textwidth}
    \centering

















\begin{tikzpicture}[>=stealth]

\node[circle, draw, minimum size=8mm] (1) at (0,2) {1};

\node[circle, draw, minimum size=8mm] (2) at (-1.2,0.5) {2};
\node[circle, draw, minimum size=8mm] (3) at ( 1.2,0.5) {3};

\node[circle, draw, minimum size=8mm] (4) at (-1.2,-0.8) {4};

\draw (-1.8,-1.4) rectangle (-0.6,1.1);

\draw (0.6,-0.1) rectangle (1.8,1.1);

\draw[dashed] (-0.6,0.5) -- node[below]{\small Nash} (0.6,0.5);

\draw[->] (1) -- (2);
\draw[->] (1) -- (3);
\draw[->] (2) -- (4);

\end{tikzpicture}
\end{minipage}
\hfill
\begin{minipage}{0.4\textwidth}
    \centering
    \scalebox{0.8}{
    \begin{minipage}{0.4\textwidth}
    \input{figs/snippet_running_example_optimization_problem}
    \end{minipage}
    }
\end{minipage}
\vspace{-10pt}
\caption{\label{fig:graph-definitions}
    We introduce three representations of a mixed-hierarchy game which models the information structure of the example from the introduction and satisfies \Cref{ass:forest-info-structure}:
    We include a set-based representation which lists elements of each robot's leader/follower/Nash sets
     (left), a graphical representation of $\graph{}$ (middle), and a coupled optimization problem (right).
    In this game, robot 1 leads all other robots (though robot 4 is led indirectly), robot 2 leads robot 4, and robot 2's subtree (including robot 4) is in a Nash relationship with robot 3.
}
\vspace{-6pt}
\end{figure*}
\begin{figure}[ht]
\centering

\begin{subfigure}{0.45\columnwidth}
\centering
\begin{tikzpicture}[>=stealth]

\node[circle, draw, minimum size=8mm] (1) at (0,1.2) {1};
\node[circle, draw, minimum size=8mm] (2) at (-1,0) {2};
\node[circle, draw, minimum size=8mm] (3) at (1,0) {3};

\draw[->] (1) -- (2);
\draw[->] (3) -- (1);
\draw[->] (2) -- (3);

\end{tikzpicture}
\caption{Hierarchy structure graph can not contain a directed cycle.}
\end{subfigure}
\hfill
\begin{subfigure}{0.45\columnwidth}
\centering
\begin{tikzpicture}[>=stealth]

\node[circle, draw, minimum size=8mm] (1) at (0,0.3) {1};
\node[circle, draw, minimum size=8mm] (2) at (1.8,0.3) {2};
\node[circle, draw, minimum size=8mm] (3) at (0.9,-0.9) {3};

\draw (-0.6,-0.3) rectangle (2.4,0.9);

\draw[->] (1) -- (3);
\draw[->] (2) -- (3);

\draw[dashed] (0.4,0.3) -- node[above]{\small Nash} (1.4,0.3);

\end{tikzpicture}
\caption{Hierarchy structure is sound but violates \Cref{ass:forest-info-structure}.}
\end{subfigure}
\vspace{-6pt}
\caption{\label{fig:invalid-graphs}Examples of invalid hierarchical information structure graphs. 
}
\vspace{-2em}
\end{figure}




\subsection{Equilibrium Conditions for Mixed-Hierarchy Games}
\label{ssec:equilibrium-conditions-in-mhgs}
Next, we define the equilibrium conditions associated with the mixed-hierarchy information structure under \Cref{ass:forest-info-structure}.


In the two-robot setting, if a robot
is a Stackelberg leader, 
then the leader communicates its strategy to the follower, who uses that knowledge when making its own decision.
The leader anticipates the followers' best response and makes a decision to steer the follower's decision to a desired outcome. 
Mathematically, this type of relationship is enforced by defining an optimal follower response map that can be embedded within the leader's problem, cf.
\cite[Ch. 3, 7]{bacsar1998dynamic}.

In an $\numrobots{}$-robot mixed hierarchy game, when robot $\idxi{}$ is an equality-constrained Stackelberg follower of (direct and indirect) leaders $\allleaderset{\idxi{}}$, the set of all optimal follower responses $\optimalresponseset{\idxi{}}{\decVar{\allleaderset{\idxi{}}}{}}$ given its leaders' communicated strategies $\decVar{\allleaderset{\idxi{}}}{}$ is 
\begin{subequations}\label{eq:stackelberg-optimal-response-set}
\begin{align}
&\hspace{-3em}\optimalresponseset{\idxi{}}{\decVar{\allleaderset{\idxi{}}}{}} \equiv 
\argmin_{\decVar{\idxi{}}{}, \decVar{\allfollowerset{\idxi{}}}} 
\costfn{\idxi{}}\!\left(
\decVar{\allleaderset{\idxi{}}}{},
\decVar{\idxi{}},
\decVar{\allfollowerset{\idxi{}}},
\decVar{\nashset{\idxi{}}}{} \right) \label{eq:optimal-response:objective} \\
\hspace{2.7em}&\text{s.t.}~~\equalityconstraint{\idxi{}}(\decVar{\idxi{}}) = 0, 
\label{eq:optimal-response:eq-constraint} \\
&\hspace{1.5em}
\left(\decVar{\idxj{}},\, \decVar{\allfollowerset{\idxj{}}}\right) \in
\optimalresponseset{\idxj{}}{\underbrace{\decVar{\idxi{}}, \decVar{\allleaderset{\idxi{}}}{}}_{\decVar{\allleaderset{\idxj{}}}}}, \quad\quad \forall j \in \followerset{i}.  \label{eq:optimal-response:optimal-subfollower-response} 
\end{align}
\end{subequations}
Eq. \cref{eq:optimal-response:objective} 
describes the minimization of robot $\idxi{}$'s objective function with respect to its own strategy and its followers' strategies (which it can influence).
Since each robot $\idxi{}$ is the leader of the robots in its own subtree $\tree{\idxi{}}$, the 
optimal follower response includes robot $\idxi{}$'s strategy $\decVar{\idxi{}*}$ and those of all of its followers, $\decVar{\allfollowerset{\idxi{}}*}$.
This best response optimization is subject to two constraints: first, the equality constraint \cref{eq:optimal-response:eq-constraint} corresponds to \cref{eq:equality-constraints}.
Second, \cref{eq:optimal-response:optimal-subfollower-response} accounts for followers' reactions to  robot $\idxi{}$'s choice of strategy, because the best responses affect robot $\idxi{}$'s objective.
Thus, \cref{eq:optimal-response:optimal-subfollower-response} recursively collects the optimal responses of robots within $\tree{\idxi{}}$ that make up $\decVar{\allfollowerset{\idxi{}}*}$.
We highlight that \cref{eq:optimal-response:objective} relies on Nash-related robots' trajectories $\decVar{\nashset{\idxi{}}}$ which robot $\idxi{}$ cannot directly influence.
Thus, each robot's equilibrium conditions must incorporate both Stackelberg leader conditions that embed \cref{eq:stackelberg-optimal-response-set} \emph{and} Nash conditions on the relationship between $\decVar{\idxi{}*}, \decVar{\allfollowerset{\idxi{}}*}$, and $\decVar{\nashset{\idxi{}}*}$.

Generally, a best response set $\optimalresponseset{\idxi{}}{\decVar{\allleaderset{\idxi{}}}{}}$ can contain multiple optimal combinations of values for $\decVar{\idxi{}*}{}$ and $ \decVar{\allfollowerset{\idxi{}}*}{}$, even in the simple two-robot quadratic case \cite[Ch. 7]{bacsar1998dynamic}.
Nevertheless, for ease of analysis, it is common to assume $|\optimalresponseset{\idxi{}}{\decVar{\allleaderset{\idxi{}}*}{}}| = 1$ \cite{bacsar1998dynamic,khan2024leadership}, i.e., that \emph{equilibrium} leader strategies result in a unique optimal response for all followers of robot $\idxi{}$.
We call a function that generates an element 
$\left(\decVar{\idxi{}*}, \decVar{\allfollowerset{\idxi{}}*} \right) \in \optimalresponseset{\idxi{}}{\decVar{\allleaderset{\idxi{}}}{}}$ a \emph{policy} $\allpolicies{\idxi{}}(\decVar{\allleaderset{\idxi{}}})$, where
\begin{equation}
\label{eq:unique-follower-policies}
 \left(\decVar{\idxi{}*}, \decVar{\allfollowerset{\idxi{}}*} \right) \equiv \allpolicies{\idxi{}}(\decVar{\allleaderset{\idxi{}}}{}) \quad  \text{and, at } \decVar{\allleaderset{\idxi{}}*}{}, \quad \optimalresponseset{\idxi{}}{\decVar{\allleaderset{\idxi{}}*}{}} = \left\{ \allpolicies{\idxi{}}(\decVar{\allleaderset{\idxi{}}*}{}) \right\}.
\end{equation}
We equivalently define $\allpolicies{\idxi{}}$ with component sub-policies $\decVar{\idxj{}*} \!=\! \policyfn{}{\idxj{}}(\decVar{\allleaderset{\idxj{}}})$ for each follower's individual best-responses, 
\begin{equation}
\label{eq:subpolicy-definition}
    \allpolicies{\idxi{}}(\decVar{\allleaderset{\idxi{}}}) := \left\{\decVar{\idxj{}*} = \policyfn{}{\idxj{}}(\decVar{\allleaderset{\idxi{}}}, \decVar{(\allleaderset{\idxj{}} \setminus \allleaderset{\idxi{}})*}) ~|~
    \idxj{} \in \{\idxi{}\} \cup \allfollowerset{\idxi{}}  \right\}.
\end{equation}
Only leader strategies $\decVar{\allleaderset{\idxi{}}}$ are provided to $\allpolicies{\idxi{}}$,
so the remaining leader strategies, $\decVar{(\allleaderset{\idxj{}} \setminus \allleaderset{\idxi{}})*}$, are computed from other subpolicies in \cref{eq:subpolicy-definition}.

\begin{definition}[Mixed-Hierarchy Equilibrium]
\label{def:mixed-hierarchy-equilibrium}
At a mixed-hierarchy equilibrium, each robot $\idxi{}$ must satisfy both Stackelberg leadership conditions to (direct and indirect) followers $\allfollowerset{\idxi{}}$ and Nash conditions with $\nashset{\idxi}$.
Given equilibrium leader strategies $\decVar{\allleaderset{\idxi{}}*}$ and equilibrium Nash-related robot strategies $\decVar{\nashset{\idxi{}}*}$, we introduce the following mixed-hierarchy equilibrium condition for each robot $\idxi{}$: 
\begin{subequations}\label{eq:mixed-hierarchy-equilibrium-condition}
\begin{align}
\left(\decVar{\idxi{}*}, \decVar{\allfollowerset{\idxi{}}*}\right) \!\in
&\argmin_{\decVar{\idxi{}}\!, \decVar{\allfollowerset{\idxi{}}}}
\costfn{\idxi{}}\!\left(
\decVar{\allleaderset{\idxi{}}*}, 
\decVar{\idxi{}}, 
\decVar{\allfollowerset{\idxi{}}},
\decVar{\nashset{\idxi{}}*}
\right) \label{eq:mixed-hierarchy-equilibrium-condition:objective}\\
\mathrm{s.t.}&~\equalityconstraint{\idxi{}}(\decVar{\idxi{}}) = 0, 
\label{eq:mixed-hierarchy-equilibrium-condition:eq-constraint} \\
\phantom{\text{s.t.}}&~
\decVar{\idxj{}}{}
\! = \! \policyfn{}{\idxj{}}(\decVar{\allleaderset{\idxi{}}*}\!\!, 
\decVar{\allleaderset{\idxj{}} 
\setminus \allleaderset{\idxi{}}}), \quad \forall \idxj{} \!\in\! \allfollowerset{\idxi{}}. \label{eq:mixed-hierarchy-equilibrium-condition:optimal-subfollower-response} 
\end{align}
\end{subequations}
\end{definition}

Condition \cref{eq:mixed-hierarchy-equilibrium-condition} incorporates both Stackelberg leadership and Nash conditions.
It requires each robot $\idxi{}$ to choose its optimal leader strategy to guide its followers $\allfollowerset{\idxi{}}$ to benefit the leading robot, as enforced by \cref{eq:mixed-hierarchy-equilibrium-condition:objective} and \cref{eq:mixed-hierarchy-equilibrium-condition:optimal-subfollower-response}. 
Simultaneously, the optimal strategies of robot $\idxi{}$ and its followers, 
$\decVar{(\{\idxi\} \cup \allfollowerset{\idxi})*}$, must be unilaterally optimal with respect to the equilibrium strategies of Nash-related robots, $\decVar{\nashset{\idxi{}}*}$, as enforced by \cref{eq:mixed-hierarchy-equilibrium-condition:objective}.
When condition 
\cref{eq:mixed-hierarchy-equilibrium-condition} 
holds for all robots $\idxi{} \in \Nset{}$ within some neighborhood of a point in the joint strategy space for all robots, $\decVar{*}{}$, then we say that $\decVar{*}{}$ is a \emph{local (equality-constrained) mixed-hierarchy equilibrium}.
\begin{assumption}[Linear-Independence Constraint Qualification]
To ensure that necessary KKT conditions hold, we require the linear independence constraint qualification (LICQ \cmt{is a standard regularity condition for these problems}), cf. \cite[Def. 12.4]{nocedal1999numerical}, to hold in \cref{eq:mixed-hierarchy-equilibrium-condition}, for each robot $\idxi{}$. 
\end{assumption}
We note that conditions \cref{eq:mixed-hierarchy-equilibrium-condition}, when aggregated across all robots, define an equality-constrained variant of the mathematical programming network problem introduced by \cite{laine2024mathematical}, under \Cref{ass:forest-info-structure}.

\newcommand{\constvelocity}[1]{v^{#1}_c}
\newcommand{\dsafe}{d_{\mathrm{s}}}
\newcommand{\colsensit}{c_\mathrm{s}}

\subsection{Running Example: $\numplayers{}$-Robot Convoy Merging Game}
\label{sec:running-example}

To make the problem formulation more concrete, consider the following  example of a game in which $\numplayers{}$ robots each select trajectories in a convoy merging scenario, taking place over
finite horizon $\horizon{} \in \mathbb{N} \setminus \{ 0 \}$, and with time discretization $\dt{} > 0$. 
At each time step $t \in [\horizon{}]$, each robot $\idxi{} \in [\numrobots{}]$ has 
state vector $\state{\idxi{}}{t} \in \mathbb{R}^{\numstates{\idxi{}}}$ and
control vector $\ctrl{\idxi{}}{t} \in \mathbb{R}^{\numctrls{\idxi{}}}$.
Concretely, each vehicle $\idxi{}$'s state contains planar horizontal and vertical position $\bm{p}^{\idxi}_t = (p^{\idxi{}}_{x,t}, ~ p^{\idxi{}}_{y,t})^\T \in \mathbb{R}^2$, heading $\psi^{\idxi{}}_t \in \mathbb{R}$, and scalar velocity $v^{\idxi{}}_t \in \mathbb{R}$, and vehicle $\idxi{}$'s controls contain longitudinal acceleration $a^{\idxi{}}_t \in \mathbb{R}$ and rotational velocity $\omega^{\idxi{}}_t \in \mathbb{R}$, i.e. $\state{\idxi{}}{t} = ( p^{\idxi{}}_{x,t}, p^{\idxi{}}_{y,t}, \psi^{\idxi{}}_{t}, v^{\idxi{}}_t )^\T, \ctrl{\idxi{}}{t} = (a^{\idxi{}}_t, \omega^{\idxi{}}_t )^\T$.
Each robot $\idxi{}$'s strategy takes the form 
\begin{equation}
\label{eq:dynamic-game-trajectory}
\decVar{\idxi{}}{} = (\state{\idxi{}\T}{1},~ \ctrl{\idxi{}\T}{1},~ \cdots~,~ \state{\idxi{}\T}{\horizon{}},~ \ctrl{\idxi{}\T}{\horizon{}})^{\T}.
\end{equation}
Robot $\idxi$'s trajectory $\decVar{\idxi}$ is constrained to start at an initial state $\initstate{\idxi{}}$ and for its state to evolve in time according to robot dynamics $\dynamics{\idxi{}}$, defined such that
\begin{equation}
    \label{eq:dynamics}
    \state{\idxi{}}{t+1} = \dynamics{\idxi{}}(\state{\idxi{}}{t}, \ctrl{\idxi{}}{t}), 
    \qquad t \in [\horizon{}-1].
\end{equation}
Concretely, for this example, we take these to encode unicycle dynamics, i.e.
\begin{equation}
\label{eq:unicycle-dynamics}
\state{\idxi{}}{t+1} = 
\dynamics{\idxi{}}(\state{\idxi{}}{t}, \ctrl{\idxi{}}{t}) 
= \left[ \begin{array}{c}
    p^{\idxi{}}_{x,t+1} \\
    p^{\idxi{}}_{y,t+1} \\
    \psi^{\idxi{}}_{t+1} \\
    v^{\idxi{}}_{t+1} \\
\end{array} \right]
= \left[ \begin{array}{c}
    p^{\idxi{}}_{x,t} + \Delta t v^{\idxi{}}_{t} \cos\psi^{\idxi{}}_t \\
    p^{\idxi{}}_{y,t} + \Delta t v^{\idxi{}}_{t} \sin\psi^{\idxi{}}_t \\
    \psi^{\idxi{}}_{t} + \Delta t \omega^{\idxi{}}_{t} \\
    v^{\idxi{}}_{t} + \Delta t a^{\idxi{}}_{t} \\
\end{array} \right].
\end{equation}
Ultimately, we can then transcribe each robot's initial state and dynamics constraints as
\begin{equation}
\label{eq:running-example-constraint}
\equalityconstraint{\idxi{}}(\decVar{\idxi{}}) = \left[ \begin{array}{c} \initstate{\idxi{}} - \state{\idxi{}}{1} \\ 
\state{\idxi{}}{2} - \dynamics{\idxi{}}({\state{\idxi{}}{1}, \ctrl{\idxi{}}{1}}) \\
\vdots \\
\state{\idxi{}}{\horizon{}} - \dynamics{\idxi{}}({\state{\idxi{}}{\horizon{}-1}, \ctrl{\idxi{}}{\horizon{}-1}})
\end{array} \right] = 0.
\end{equation}
Additional equality constraints---e.g., encoding a desired terminal state---can be appended to \cref{eq:running-example-constraint} as needed.

Lastly, we construct convoy merging scenarios in which  $\numrobots{}_c$ vehicles  have formed a convoy on the road, and $\numrobots{}_m = \numrobots{} - \numrobots{}_c$ vehicles wish to merge.
All vehicles wish to track the center of the merged lane (with their vertical positions and headings), maintain a consistent speed $\constvelocity{\idxi{}}$,
minimize their control effort, and avoid collision.
The objective, given below, is parameterized by positive weights $\bm{w}^\idxi{} = (w^{\idxi, 1}, w^{\idxi, 2}, w^{\idxi, 3}, w^{\idxi, 4}, w^{\idxi, 5})$, collision safe distance $\dsafe > 0$ beyond which the collision cost is 0, and collision sensitivity $\colsensit > 0$: 
\begin{align}
\label{eq:merge-cost}
\hspace{-0.7em}\costfn{\idxi{}}(\allDecVars{}) 
&= 
\sum_{t=1}^{\horizon{}} \Bigg(
w^{\idxi, 1}\|\ctrl{\idxi{}}{t}\|_2^2
+ w^{\idxi, 2} \big(p^{\idxi{}}_{y,t}\big)^{2}
+ w^{\idxi, 3} \big(\psi^{\idxi{}}_t\big)^2 
+ w^{\idxi, 4}\big(v^{\idxi{}}_t - \cmt{\constvelocity{\idxi{}}}\big)^{2} \\ 
&+ w^{\idxi, 5} \sum_{\idxj=1}^{\numrobots{}-1}
  \sum_{\idxk=\idxj+1}^{\numrobots{}} 
  \left[\frac{1}{\colsensit}\log\left(1+\exp\big(\colsensit\, (\dsafe^2 - \big\|\bm{p}^{\idxj}_{t} - \bm{p}^{\idxk}_{t}\big\|_2^2 )\big)\right)\right]^2 \Bigg) \notag{}
\end{align}
The merging vehicles start off the merging lane and are incentivized to track the center of the round entry ramp
and merge safely onto the highway, as shown in \cref{fig:merging-experiments}. 
We enforce curve-following by altering \cref{eq:merge-cost} for vehicle 3 to include a component term, $(\|\bm{p}^{\idxi}_t\|_2^2 - R^2)^2$, and we alter vehicle 1's cost to include a component term $(p^{1}_{x,t} - p^{3}_{x,t} - \mathrm{d}_{13})^2$, which incentivizes vehicle 1 to stay $\mathrm{d}_{13}$ meters ahead of vehicle 3.
In \cref{ssec:experiments-merging-scenario},
we experiment with different hierarchical structures (e.g., Nash, two mixed hierarchies, and Stackelberg chain) and report their impact on vehicles' decisions.

\section{Solving Mixed-Hierarchy Games}
\label{sec:method}



In order to build an algorithm to solve mixed-hierarchy games (i.e., find a set of robot strategies that satisfy \Cref{def:mixed-hierarchy-equilibrium}), we must first develop a
procedure for constructing appropriate Karush-Kuhn-Tucker (KKT) conditions for first-order optimality.
As we will show, these conditions involve high-order derivatives of follower robots' policies; to circumvent the associated computational challenges, we describe a \emph{quasi-policy} approximation, which embeds local first-order approximations of followers' policies within each leader's optimization problem in order to reduce the recursive complexity of evaluating high-order derivatives.
Finally, we propose an inexact Newton's method to iteratively solve the approximate KKT system and prove local exponential convergence.

\subsection{Backward Induction for Deriving KKT Conditions}

We begin by constructing KKT conditions for robots without followers (leaves in the graph $\graph{}$) and then proceed in reverse topological order.

\subsubsection{Constructing KKT Conditions for Leaves.}
Equilibrium conditions for robots without followers do not include the follower constraints \cref{eq:mixed-hierarchy-equilibrium-condition:optimal-subfollower-response}.
Thus, a leaf robot $\idxi{}$'s optimization problem becomes
\begin{equation}
\label{eq:leaf-player-problem}
\costfn{\idxi{}}\left( \allDecVars{}^* \right) = \min_{\decVar{\idxi{}}}
\costfn{\idxi{}}\!\left(
\decVar{\allleaderset{\idxi{}}*}, 
\decVar{\idxi{}}, 
\decVar{\nashset{\idxi{}}*}
\right)
\quad \text{s.t.}~\equalityconstraint{\idxi{}}(\decVar{\idxi{}}) = 0.
\end{equation}
We associate the equality constraint $\equalityconstraint{\idxi{}}(\decVar{\idxi{}}) = 0$ with a Lagrange multiplier $\constraintDual{\idxi{}} \in \mathbb{R}^{\numequalitylagrange{\idxi{}}}$, and introduce the corresponding Lagrangian:
\begin{equation}
\label{eq:leaf-lagrangian}
\lagrangian{\idxi{}}(\allDecVars{}, \constraintDual{\idxi{}}) = \costfn{\idxi{}}(\allDecVars{}) - \constraintDual{\idxi{}}{}^\T \equalityconstraint{\idxi{}}(\decVar{\idxi{}}) . 
\end{equation}
We then compactly formulate robot $\idxi$'s KKT conditions $\kktConds{\idxi}$ by concatenating its Lagrangian stationarity condition with the original equality constraint, 
\begin{equation}
\label{eq:leaf-kkt-conditions}
\kktConds{\idxi{}}(\allDecVars{}, \constraintDual{\idxi{}})
=\begin{bmatrix} \nabla_{\decVar{\idxi{}}}\lagrangian{\idxi{}}(\allDecVars{}, \constraintDual{\idxi{}})\\
\equalityconstraint{\idxi{}}(\decVar{\idxi{}})
\end{bmatrix}=0.
\end{equation}


\subsubsection{From Leaf Robots' KKT Conditions to Their Policies.}
Consider a robot $\idxj{}$ with (direct or indirect) leader robot $\idxi{} \in \allleaderset{\idxj{}}$, but no followers.
Robot $\idxj{}$'s KKT conditions $\kktConds{\idxj{}}$ implicitly specify its \textit{policy} $\policyfn{}{\idxj{}}$, which maps leaders' strategies $\decVar{\allleaderset{\idxj{}}}$ to robot $\idxj{}$'s best response $\decVar{\idxj{}}$, i.e. so that $\decVar{\idxj{}} = \policyfn{}{\idxj{}}(\decVar{\allleaderset{\idxj{}}})$.
For notational convenience, we 
denote the primal and dual variables making up the leaf robot's decision variables by $\outvar{\idxj} = (\decVar{\idxj{}\T}, \constraintDual{\idxj{}\T})^\T$, and 
compactly represent all variables relevant to leaf robot $\idxj$'s optimization problem by 
$\totalvar{\idxj} = (\decVar{\allleaderset{\idxj}\T}, \outvar{\idxj\T})^\T$.

Before constructing the leader $\idxi$'s KKT conditions $\kktConds{\idxi}$, we must first understand how follower $\idxj$'s best response strategy $\decVar{\idxj}$ locally changes in response to small changes in the leader's decision.
To this end, we express the first-order Taylor approximation of robot $\idxj$'s KKT conditions as follows:
\begin{equation}
    \label{eq:const-rank-thm}
    0 = \kktConds{\idxj}(\totalvar{\idxj}) + \nabla_{\outvar{\idxj}} \kktConds{\idxj}(\totalvar{\idxj}) \cdot \Delta \outvar{\idxj} + \nabla_{\decVar{\allleaderset{\idxj}}} \kktConds{\idxj}(\totalvar{\idxj}) \cdot \Delta \decVar{\allleaderset{\idxj}},
\end{equation}
where $\Delta \decVar{\allleaderset{\idxj}}$ is a small perturbation in the leaders' strategies and $\Delta \outvar{\idxj}$ is the corresponding change in $\idxj$'s best response (including both primal and dual variables).
We note that $\nabla_{\outvar{\idxj}} \kktConds{\idxj}(\totalvar{\idxj})$ is a square matrix; if it has full rank, we can apply the implicit function theorem \cite{krantz2002implicit} and obtain that, locally:
\begin{equation}
    \label{eq:best-response-policy}
    \Delta\outvar{\idxj}=\allpolicies{\idxj}( \Delta \decVar{\allleaderset{\idxj}} ) =
    - \big(\nabla_{\outvar{\idxj}} \kktConds{\idxj}(\totalvar{\idxj})\big)^{-1}(\kktConds{\idxj}(\totalvar{\idxj}) + \nabla_{\decVar{\allleaderset{\idxj}}} \kktConds{\idxj}(\totalvar{\idxj}) \cdot \Delta \decVar{\allleaderset{\idxj}}).
\end{equation} 
From \cref{eq:best-response-policy}, we extract the policy $\decVar{\idxj}=\policyfn{}{\idxj}(\decVar{\allleaderset{\idxj}}) $ for robot $\idxj$'s strategy by selecting the corresponding rows of $\allpolicies{\idxj}( \Delta \decVar{\allleaderset{\idxj}} )$.
This explicit construction can then replace the implicit constraint \cref{eq:mixed-hierarchy-equilibrium-condition:optimal-subfollower-response} in the optimization problem of each leader $\idxi{} \in \allleaderset{\idxj}$, which will shortly allow us to build an explicit expression for the leader's approximate KKT conditions and ultimately an efficient solver.

\subsubsection{Constructing KKT Conditions for Non-Leaf Robots.}
Proceeding in reverse topological order on graph $\graph$ from a leaf robot, we consider a robot $\idxi$ whose followers $\allfollowerset{\idxi}$ are all leaves.
Let robot $\idxj \in \allfollowerset{\idxi}$ and denote its local policy by $\policyfn{}{\idxj}$, per the construction in \cref{eq:leaf-kkt-conditions}.
As discussed above, we replace the implicit constraints \cref{eq:mixed-hierarchy-equilibrium-condition:optimal-subfollower-response} on  followers' best responses in robot $\idxi$'s problem  \cref{eq:mixed-hierarchy-equilibrium-condition} with an explicit approximation constructed from \cref{eq:best-response-policy}.

To this end, we introduce dual variable $\constraintDual{\idxi} \in \mathbb{R}^{\numequalitylagrange{\idxi}}$ associated with the equality constraint \cref{eq:mixed-hierarchy-equilibrium-condition:eq-constraint} and $\policyDual{\idxi}{\idxj} \in \mathbb{R}^{\numpolicyduals{\idxj}}$ corresponding to the best-response constraint \cref{eq:mixed-hierarchy-equilibrium-condition:optimal-subfollower-response}, for each $\idxj\in\allfollowerset{\idxi}$.
We compactly represent all policy duals associated with robot $\idxi$ as $\allFollowerPolicyDuals{\idxi} = ( \policyDual{\idxi{}}{\idxj{}\T} ~|~ \idxj{} \in \allfollowerset{\idxi{}} )^\T$\!\!, and collect all decision variables for robot $\idxi$ as
$\outvar{\idxi} \!=\! (\decVar{\idxi{}\T} \!, \decVar{\allfollowerset{\idxi{}}\T}\!, \constraintDual{\idxi{}\T}\!, \allFollowerPolicyDuals{\idxi{}\T})^\T$\!\!.
We also collect all variables relevant to robot $\idxi$'s optimization problem as $\totalvar{\idxi} \!=\! (\decVar{\allleaderset{\idxi{}}\T}\!, \outvar{\idxi{}\T})^\T$\!\!.
We then write robot $\idxi$'s Lagrangian
\begin{equation}
\label{eq:lagrangian}
\lagrangian{\idxi{}}(\totalvar{\idxi}) = \costfn{\idxi{}}(\allDecVars{}) - \constraintDual{\idxi{}}{}^\T \equalityconstraint{\idxi{}}(\decVar{\idxi{}}) - \!\!\sum_{\idxj{} \in \allfollowerset{\idxi{}}} \policyDual{\idxi{}}{\idxj{}}{}^\T (\decVar{\idxj{}} - \policyfn{\idxi}{\idxj{}}(\decVar{\allleaderset{\idxj{}}})). 
\end{equation}
Finally, we define robot $\idxi$'s KKT conditions by concatenating Lagrangian stationarity conditions for robot $\idxi{}$ and its followers, its followers' best-response constraints, and its own equality constraints,
\begin{equation}
\label{eq:kkt-conditions}
\kktConds{\idxi{}}(\totalvar{\idxi{}})
=\begin{bmatrix}
    \nabla_{\decVar{\idxi{}}}\lagrangian{\idxi{}}(\totalvar{\idxi{}}) \\
    \begin{bmatrix} \nabla_{\decVar{\idxj{}}}\lagrangian{\idxi{}}(\totalvar{\idxi{}}) \\ 
    \decVar{\idxj} - \policyfn{}{\idxj}(\decVar{\allleaderset{\idxj}}) \end{bmatrix}_{\idxj{}\in\allfollowerset{\idxi{}}} \\
    \equalityconstraint{\idxi{}}(\decVar{\idxi{}})
\end{bmatrix} 
= 0.
\end{equation}
Evaluating a direct follower $\idxj$'s stationarity condition 
in \cref{eq:kkt-conditions} requires computing $\nabla_{\decVar{\idxj{}}}\lagrangian{\idxi{}}(\totalvar{\idxi{}}) $, which depends on the policy gradient $\nabla_{\decVar{\idxj{}}} \policyfn{}{\idxj{}} $:  
\begin{equation}
    \label{eq:quasi-policy-approximation}
    \nabla_{\decVar{\idxj{}}}\lagrangian{\idxi{}}(\totalvar{\idxi{}}) = \nabla_{\decVar{\idxj{}}}\costfn{\idxi{}}(\allDecVars{}) 
    - \!\!\sum_{\idxj{} \in \allfollowerset{\idxi{}}}  (I - \nabla_{\decVar{\idxj}} \policyfn{\idxi}{\idxj{}}(\decVar{\allleaderset{\idxj{}}}))^\T\policyDual{\idxi{}}{\idxj{}}{} .
\end{equation}
\cmt{By the same mechanism, a depth-3 chain (robot $\idxi$ leading $\idxj$ leading $\idxk$) introduces the second-order policy derivative $\nabla^2_{\decVar{\idxk}} \policyfn{}{\idxk}$ into  robot $\idxi$'s KKT system via differentiation of $\policyfn{}{\idxj}$ (which itself depends on $\nabla_{\decVar{\idxj}} \policyfn{}{\idxj}$).
The order of policy gradients increases with hierarchy depth, introducing tractability challenges.}
To make such computation tractable, we take inspiration from \cite{laine2023computation,li2024computation} and propose to replace leaders' exact optimality conditions with an approximate set of KKT conditions where we locally keep the first-order policy gradient (for leader-follower relationships) and zero out higher-order policy gradients wherever they appear in robot $\idxi{}$'s KKT conditions. 
This \emph{quasi-policy} approximation makes it straightforward to recursively continue the aforementioned procedure in reverse topological order.
Once each of robot $\idxi{}$'s followers' quasi-policy approximations is computed, it can be embedded in robot $\idxi$'s problem, even if the followers are not leaves.

\subsubsection{Encoding Nash Relationships}
Every robot's optimal choice of strategy depends on its leaders', followers', and Nash relatives' strategies.
While the KKT construction process described above encodes leader-follower relationships, Nash relationships are captured by solving KKT conditions jointly across sets of multiple robots, i.e. by concatenating KKT conditions of robots with the same leaders.

\subsubsection{Constructing Joint KKT Conditions.}
Ultimately, the full set of KKT conditions for an equality-constrained mixed-hierarchy equilibrium 
has size $\numtotalvars{}$ (the number of variables match the number of nonlinear equations) and is given by
\begin{equation}
\vspace{-0.5em}
\label{eq:full-kkt-conditions}
\kktConds{}(\totalvar{})
=
\Big[ \kktConds{1}(\totalvar{1})^\T, ~~ \kktConds{2}(\totalvar{2})^\T, ~ \cdots, ~ \kktConds{\numplayers{}}(\totalvar{\numplayers{}})^\T\Big]^\T = 0.
\vspace{-0.5em}
\end{equation}


\subsection{Inexact Newton's Method for Solving Mixed Hierarchy Games}











\begin{algorithm}[!b]
\LinesNumbered
\caption{General Solver for \MethodName{}s}
\label{alg:nonquadratic-solver}
\DontPrintSemicolon
\KwIn{Hierarchy structure $\graph{}$, initial guess $\totalvariter{}{0}$, tolerance $\threshold{} > 0$, maximum iterations $\maxiters{} > 0$.}
\KwOut{Approximate solution $\totalvar{}{}^{\ast}$ satisfying first-order optimality conditions.}

\For{$\iter \gets 1$ \KwTo $\maxiters{}$}{
    \hspace{-2em}\algcmt{Construct approx. KKT conditions in reverse topological order.}
    \For(){each robot $\idxi{}$ in reverse topological order of $\graph$}{
        Construct robot $\idxi$'s KKT conditions $\kktConds{\idxi}$ per Eqs. \cref{eq:leaf-kkt-conditions,eq:kkt-conditions,eq:leaf-lagrangian,eq:lagrangian,eq:const-rank-thm,eq:best-response-policy}. \label{alg:nonquadratic-solver:construct-kkt-i}
    }
    
    Assemble and evaluate $\kktConds{}(\totalvariter{}{\iter})$ \cref{eq:full-kkt-conditions}, and solve for inexact Newton step $\Delta\totalvariter{}{\iter}$. \label{alg:nonquadratic-solver:newton-step}
    
    Select step size $\alpha_{\iter}$ along direction $\Delta \totalvariter{}{\iter}$ and generate $\totalvariter{}{\iter+1}$ per \cref{eq:newton-update}\; \label{alg:nonquadratic-solver:update}
    \If(){$\| \kktConds{}(\totalvariter{}{\iter+1}) \|_2^2 \leq \threshold{}$\label{alg:nonquadratic-solver:convergence}}{
        \Return $\totalvariter{}{\iter+1}$\;
    }
}

\Return $\totalvariter{}{\maxiters{}}$\;
\end{algorithm}

We propose to solve the joint KKT system \cref{eq:full-kkt-conditions} via an inexact Newton's method (\Cref{alg:nonquadratic-solver}). 
At iteration $\iter$, we first construct the KKT system $\kktConds{}(\totalvariter{}{\iter})$ under the quasi-policy approximation.
Then, for computational tractability, we compute an inexact Newton step for $\kktConds{}(\totalvariter{}{\iter})$ using an approximate Jacobian (denoted $\nabla_{\totalvar{}}\kktConds{}(\totalvariter{}{\iter})$), which we construct under the quasi-policy approximation by setting the high-order policy gradient terms to zero as discussed earlier.
Accordingly, we form an approximate first-order Taylor expansion of the KKT conditions about the current iterate $\totalvariter{}{\iter}$ (\cref{alg:nonquadratic-solver:newton-step}), resulting in the inexact Newton system 
\begin{equation}
\label{eq:newton-system}
\nabla_{\totalvar{}}\kktConds{}(\totalvariter{}{\iter})\, \Delta\totalvariter{}{\iter}
=
-\kktConds{}(\totalvariter{}{\iter}).
\end{equation}
Solving \cref{eq:newton-system} produces an inexact Newton direction $\Delta\totalvariter{}{\iter}$, which can be used to update the current iterate $\totalvariter{}{\iter}$ (\cref{alg:nonquadratic-solver:update}) as follows:
\begin{equation}
\label{eq:newton-update}
\totalvariter{}{\iter+1}
=
\totalvariter{}{\iter}
-
\stepsize{\iter} \Delta\totalvariter{}{\iter},
\end{equation}
\noindent with step size $\stepsize{\iter}\in(0,1]$ selected (\cref{alg:nonquadratic-solver:update}) using a standard backtracking line search, geometric decrease, etc. (cf. \cite[Ch.~3]{nocedal1999numerical}).
Iterations terminate when $\|\kktConds{}(\totalvariter{}{\iter})\|_2^2 \le \tau$ (a user-specified tolerance for KKT error; \cref{alg:nonquadratic-solver:convergence}), yielding an approximate mixed-hierarchy equilibrium; it is not exact due to the use of the quasi-policy approximation in constructing the KKT conditions.

In practice, we can accelerate mixed-hierarchy KKT computation (\cref{alg:nonquadratic-solver:construct-kkt-i}) under the quasi-policy approximation by caching reusable symbolic components (generated via a symbolic algebra toolchain \cite{SymbolicTracingUtils2025}), exploiting sparsity induced by the hierarchy graph, etc.
\cmt{Next, we comment on the quality of the approximation (\cref{rmk:approximation-quality}) and present a local convergence analysis of our method (\Cref{thm:convergence-analysis}).}


\begin{remark}[Quality of Quasi-Policy Approximations]
\label{rmk:approximation-quality}
\cmt{Provided quadratic objectives and linear constraints, the optimal policy is linear \cite{laine2024mathematical}, and the quasi-policy approximation is exact.}
In this case, our method converges in a single iteration \cmt{for} step size $\alpha = 1$. For games with nonquadratic objectives and nonlinear constraints, the quasi-policy provides a good approximation when higher-order policy gradients of the optimal policy are small relative to first-order gradients,
as suggested by analyses of both feedback Nash \cite{laine2023computation} and Stackelberg games \cite{li2024computation}.  
\cmt{On an equality-constrained Stackelberg-chain example problem from \cite{li2024computation}, our solver reproduces the equilibrium of that paper's specialized solver.}
\end{remark}

\begin{theorem}
\label{thm:convergence-analysis}
Let $\totalvar{}_{0}$ be the initialization of \cref{alg:nonquadratic-solver} and define the sub-level set $\mathcal{F}
:=
\{\totalvar{}:\|\approxkktConds{}(\totalvar{})\|_2
\le
\|\approxkktConds{}(\totalvar{}_{0})\|_2\}.$ Denote by $\nabla \approxkktConds{}(\totalvar{})$ the approximate Jacobian under the quasi-policy assumption, i.e., setting the high-order policy gradient terms to zeros.  
Suppose that $\mathcal{F}$ is closed, that $\nabla \approxkktConds{}(\totalvar{})$ is invertible for all $\totalvar{}\in\mathcal{F}$, and that there exist constants $D,C>0$ such that
$\|(\nabla \approxkktConds{}(\totalvar{}))^{-1}\|_2 \le D$, $
\|\nabla^* \approxkktConds{}(\totalvar{})-\nabla^* \approxkktConds{}(\tilde{\totalvar{}})\|_2
\le
C\|\totalvar{}-\tilde{\totalvar{}}\|_2,\ 
\forall \totalvar{},\tilde{\totalvar{}}\in\mathcal{F}$. 
Denote by $\nabla^* \approxkktConds{}(\totalvar{})$ the exact Jacobian of $\approxkktConds{}$, and assume that there exists $\delta>0$ such that $
\|\nabla^* \approxkktConds{}(\totalvar{})-\nabla \approxkktConds{}(\totalvar{})\|_2\le\delta$, $
\forall \totalvar{}\in\mathcal{F}$, and $
D\cdot \delta<1$. 
Let $
\Delta \totalvar{}
=
-(\nabla \approxkktConds{}(\totalvar{}))^{-1}
\approxkktConds{}(\totalvar{})$ 
be the inexact Newton step in Alg.~\ref{alg:nonquadratic-solver}. Then, for the $\iter{}$-th iteration solution $\totalvar{}_{k}\in\mathcal{F}$, there exists $\alpha\in[0,1]$ such that
\begin{enumerate}
\item
if $\|\approxkktConds{}(\totalvar{}_k)\|_2>\frac{1-D\delta}{D^2C}$, then $
\|\approxkktConds{}(\totalvar{}_{k}+\alpha\Delta\totalvar{}_k)\|_2
\le
\|\approxkktConds{}(\totalvar{}_k)\|_2-\frac{1-D\delta}{2D^2C}$;
\item
if $\|\approxkktConds{}(\totalvar{}_k)\|_2\le\frac{1-D\delta}{D^2C}$, then $
\|\approxkktConds{}(\totalvar{}_k+\alpha\Delta\totalvar{}_k)\|_2
\le
\frac{1}{2}(1+D\delta)\|\approxkktConds{}(\totalvar{}_k)\|_2$. 
\end{enumerate}
\end{theorem}

\begin{proof}
    By the fundamental theorem of calculus, we have \vspace{-0.5em}
    \begin{equation}
    \begin{aligned}
        &\|\approxkktConds{}(\totalvar{}+ \alpha \Delta \totalvar{})\|_2=\left\| \approxkktConds{}(\totalvar{}) + \int_0^1 \nabla^* \approxkktConds{}(\totalvar{} + \tau \alpha \Delta \totalvar{} )\alpha \Delta \totalvar{} d\tau  \right\|_2\\[-3pt]
        &\le \| \approxkktConds{}(\totalvar{}) + \alpha \nabla^*\approxkktConds{}( \totalvar{}) \Delta \totalvar{}\|_2 + \left\| \int_0^1 (\nabla^*\approxkktConds{}(\totalvar{} + \tau \alpha \Delta \totalvar{}) - \nabla^* \approxkktConds{}(\totalvar{}))\alpha \Delta \totalvar{}d\tau \right\|_2\\[-5pt]
        & \le \| \approxkktConds{}(\totalvar{}) - \alpha \nabla^* \approxkktConds{}(\totalvar{})(\nabla \approxkktConds{}(\totalvar{}))^{-1} \approxkktConds{}(\totalvar{}) \|_2 + \|\alpha \Delta \totalvar{}\|_2  \cdot \int_0^1 C \|\alpha \tau \Delta \totalvar{} \|_2 d\tau \\[-5pt]
        & \le (1-\alpha) \|\approxkktConds{}(\totalvar{})\|_2 + \alpha \delta D \|\approxkktConds{}(\totalvar{})\|_2 + \|\alpha \Delta \totalvar{}\|_2  \cdot \int_0^1 C \|\alpha \tau \Delta \totalvar{} \|_2 d\tau\\[-5pt]
        & \le (1-\alpha(1-\delta D)) \|\approxkktConds{}(\totalvar{})\|_2 + \frac{1}{2} \alpha^2 CD^2 \|\approxkktConds{}(\totalvar{})\|_2^2.
    \end{aligned}\label{eq:basic inequality}
    \end{equation}
    The right-hand side is a quadratic function of $\alpha$ and is minimized \cmt{by}  $\alpha^*=\frac{1-\delta D}{D^2 C \|\approxkktConds{}(\totalvar{}) \|_2}$. This quadratic term also ensures that for all $\alpha\in[0,\alpha^*]$, we have $\totalvar{}+\alpha\Delta\totalvar{} \in \mathcal{F}$, provided $\totalvar{} \in \mathcal{F}$. Suppose that $\|\approxkktConds{}(\totalvar{})\|_2> \frac{1-\delta D}{D^2C}$, then we have $\alpha^*<1$ is a feasible stepsize and therefore $ \| \approxkktConds{}(\totalvar{} + \alpha^* \Delta \totalvar{})\|_2 \le\|\approxkktConds{}(\totalvar{})\|_2 - \frac{(1-\delta D)^2}{2D^2C}$. Moreover, when $\|\approxkktConds{}(\totalvar{})\|_2\le \frac{1-\delta D}{D^2C}$, the right handside of \eqref{eq:basic inequality} is minimized by the maximum allowable step size $\alpha = 1$, so we can strengthen the bound by substituting $ \|\approxkktConds{}(\totalvar{})\|_2\le \frac{1-\delta D}{D^2C}$ to find $\| \approxkktConds{}(\totalvar{}+\alpha \Delta \totalvar{})\|_2\le \frac{1}{2}(1+\delta D) \|\approxkktConds{}(\totalvar{})\|_2$. 
\end{proof}
\begin{remark}[Implications of Theorem~\ref{thm:convergence-analysis}]
Theorem~\ref{thm:convergence-analysis} clarifies structural conditions under which \cref{alg:nonquadratic-solver} is expected to converge locally.
In particular, when robot objectives include strong quadratic control regularization and the dynamics are locally close to linear, higher-order policy sensitivities are small, making the assumptions of Theorem~\ref{thm:convergence-analysis} well justified.
This observation provides practical guidance for modeling objectives and constraints to promote convergence of our solver in mixed-hierarchy games.
\cmt{The constants in Theorem~\ref{thm:convergence-analysis} are regularity conditions on the problem and need not be computed by the user.}
\end{remark}

\paragraph{Computational Complexity.} \cmt{The total runtime complexity is $\mathcal{O}(\maxiters{} \cdot (\numplayers{}^2 + H(\numtotalvars{})))$, where $\numplayers{}^2$ accounts for worst-case (Stackelberg chain) KKT construction and $H(\numtotalvars{})$ for the runtime of the linear system solver for the KKT system.}
\section{Experiments}
\label{sec:experiments}
We demonstrate our method's efficacy in  
\cmt{hardware and simulated experiments}.
\subsection{Target Guarding: A Three-Robot \cmt{Hardware} Experiment}
\label{ssec:hardware-scenario}
We describe a three-robot ($\numplayers{} = 3$) dynamic mixed-hierarchy game between a pursuer (R1), a guard (R2), and a target (R3).
The pursuer (R1) wants to achieve a similar position to the target (R3), but wishes to stay away from the guard (R2).
The guard (R2) wants to keep the pursuer (R1) away from the target (R3), and the target wants to stay close to the guard (R2) while reaching a goal state $\state{}{g} \in \mathbb{R}^{\numstates{3}}$ by the end of the planning horizon.
All robots have a control effort cost, and component terms 
are weighted by $\bm{w}^{\idxi} \in \mathbb{R}^3$.
\begin{subequations}
\begin{align*}
\costfn{1}(\allDecVars{}) &= \sum_{t=1}^{\horizon{}} \bm{w}^1 \cdot \left[\| \state{3}{t} - \state{1}{t} \|_2^2, ~~ -\!\| \state{2}{t} - \state{1}{t} \|_2^2, ~~ \|\ctrl{1}{t}\|_2^2\right]^\T\!\!\!, ~ \bm{w}^1 = [2, 1, 1.25]^\T, \\
\costfn{2}(\allDecVars{}) &= \sum_{t=1}^{\horizon{}} \bm{w}^2 \cdot \left[\| \state{3}{t} - \state{2}{t} \|_2^2, ~~ -\!\| \state{3}{t} - \state{1}{t} \|_2^2, ~~ \|\ctrl{2}{t}\|_2^2\right]^\T\!\!\!, ~ \bm{w}^2 = [0.5, 1, 0.25]^\T, \\
\costfn{3}(\allDecVars{}) &= \sum_{t=1}^{\horizon{}} \bm{w}^3 \cdot \left[\frac{1}{\horizon{}}\| \state{3}{\horizon{}} - \state{}{g} \|_2^2, ~~ \| \state{3}{t} - \state{2}{t} \|_2^2, ~~ \|\ctrl{3}{t}\|_2^2\right]^\T\!\!\!, ~ \bm{w}^3 = [10, 1.25, 0.1]^\T\!\!.
\end{align*}
\end{subequations}
We define hierarchy structure $\graph{} = ([3], \{(2, 1)\})$ to 
model the guard leading the target, while both have a Nash relationship with the pursuer.

\begin{figure}[!t]
\centering
 \includegraphics[width=0.9\columnwidth]{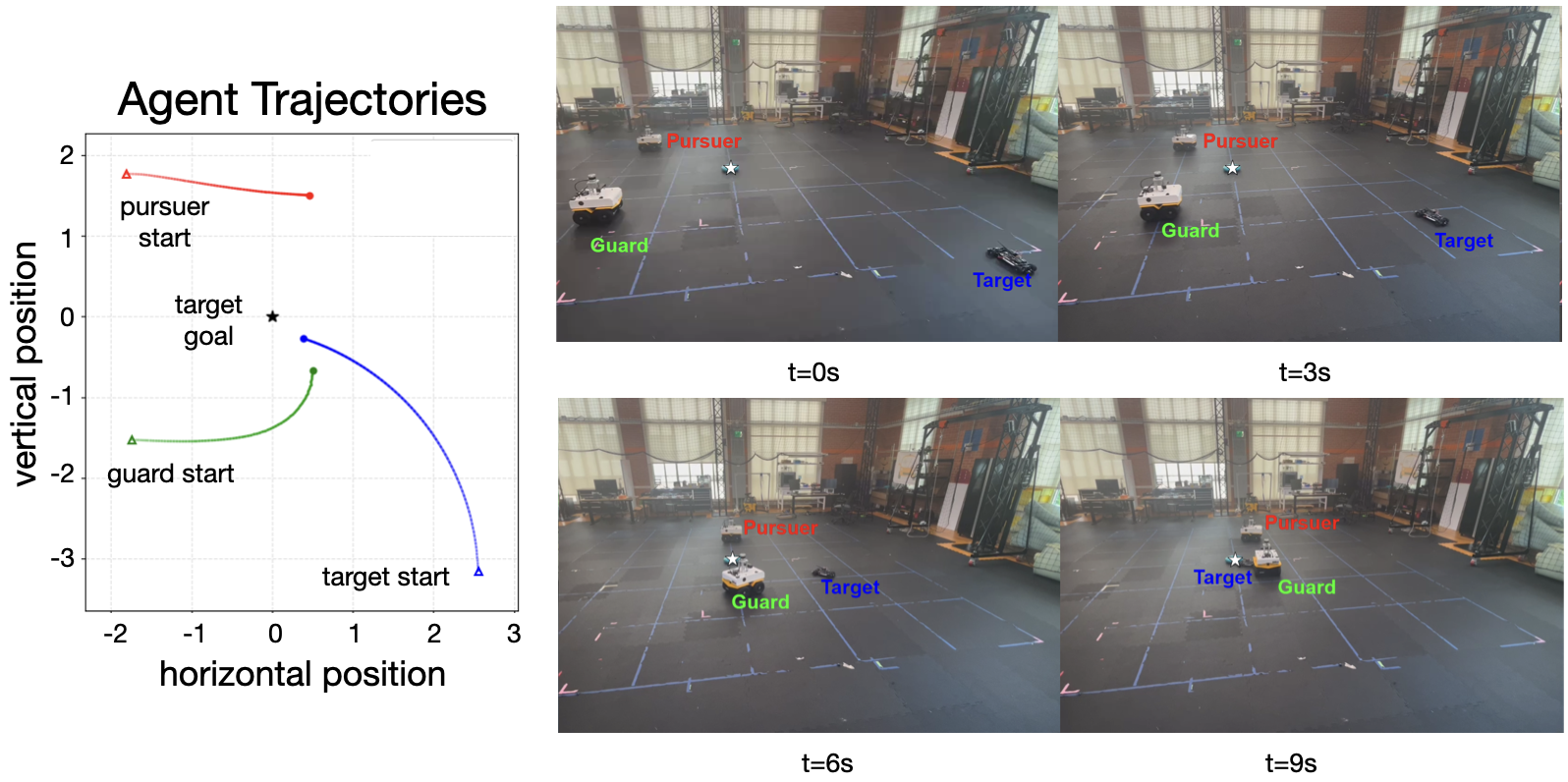}
\vspace{-6pt}
\caption{\label{fig:front-figure}
   We \cmt{play} a target-guarding game in \cmt{hardware} with \cmt{three robots}.
   Robot 2 (Clearpath Jackal, green) acts as the guard, leading Robot 1 (Clearpath Jackal, pursuer, red) and Robot 3 (NVIDIA JetRacer, target, blue), which seeks a goal. All robots execute plans generated by our solver in a receding-horizon fashion.
}
\vspace{-18pt}
\end{figure}




\subsubsection{Experimental Setup}
We evaluate this example in \cmt{hardware}
with three ground robots: two agents \cmt{(Clearpath Jackals)} acting as the pursuer (R1) and guard (R2), and \cmt{an NVIDIA JetRacer as the} target (R3). The interaction is modeled with horizon $\horizon{} = 10$ and sampling period $\dt{} = \SI{0.1}{\second}$, and solved\footnote{CPU version: 13th Gen Intel Core i9-13900HX (2.20 GHz). Memory: 32 GB.}
in a receding-horizon manner using simulator-provided state measurements.
We convert the resulting double-integrator controls to differential drive controls for each robot.

We highlight a number of noteworthy behaviors in \cref{fig:front-figure}.
R3 makes a wide turn to move towards the goal and get closer to R2.
R1 pursues a course to meeting R3, but diverts from it due to the presence of R2.
R2, using knowledge of the target's actions, chooses to intersect with R1 and R3's paths towards the origin, leading R1 to continue its course and R3 diverting its course to avoid R2.
The solver is called 93 times during execution and each call takes \SI{13.6 \pm 2.8}{\milli\second} (mean $\pm$ standard deviation).
This demonstration shows that our solver runs in real-time, making it practical for real-world deployment.


\subsection{Merging into a Convoy: A Four-Vehicle Simulation}
\label{ssec:experiments-merging-scenario}
We run the experiments described in the running example from \cref{sec:running-example} on a game 
with horizon $\horizon{}=10$, and discretization period $\Delta t =$ \SI{0.5}{\second}.
We choose the following hierarchy structures: Nash, \cmt{mixed hierarchy A $(1 \rightarrow 2, 2 \rightarrow 4)$, mixed hierarchy B $(1 \rightarrow 3, 1 \rightarrow 2, 2 \rightarrow 4)$, and Stackelberg chain (1 $\rightarrow$ 3 $\rightarrow$ 2 $\rightarrow$ 4) as depicted in \cref{fig:graph-definitions}}.
The first hierarchy structure defines simultaneous play among all vehicles (Nash).
\cmt{The second structure describes a situation in which the entire convoy is a Stackelberg chain led by vehicle 1, and they play a Nash game with vehicle 3.
In the third structure, vehicle 1 plays first and then its direct followers 2 and 3 play simultaneously in response.
This could simulate a real-world scenario in which the merging vehicle responds to leader vehicle 1's strategy
but has not yet resolved whether to yield to vehicle 2.
The fourth describes turn-based play among vehicles (Stackelberg chain).}



\begin{figure*}[!t]
\centering
\begin{minipage}{0.55\textwidth}
  \includegraphics[width=\linewidth]{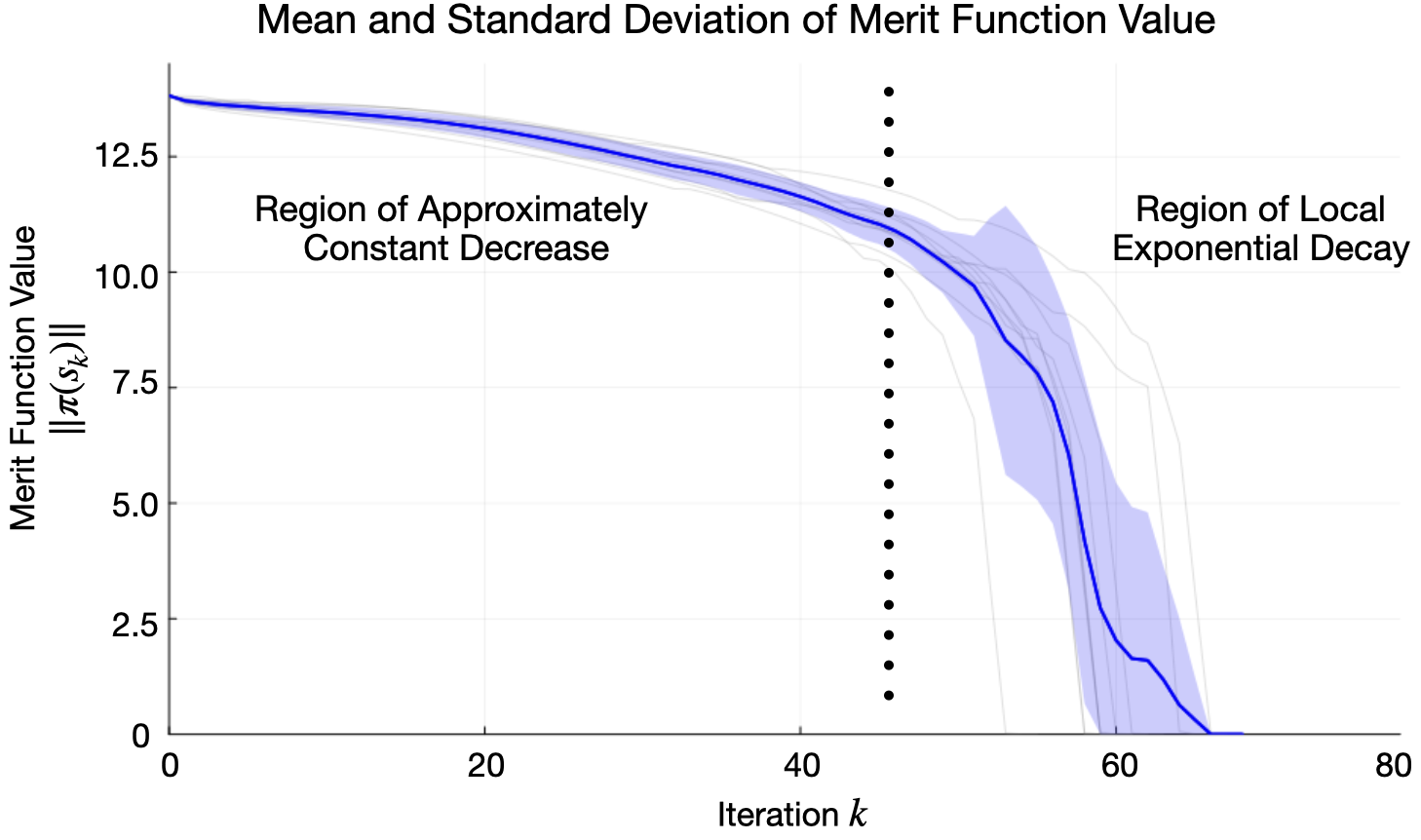}
\end{minipage}\hfill
\begin{minipage}{0.35\textwidth}
  \caption{
  We report the mean and standard deviation of the merit function value $\|\kktConds{}(\totalvariter{}{\iter})\|_2$ over 10 runs with random initial states under a fixed hierarchy structure, corroborating Theorem~\ref{thm:convergence-analysis} via monotone decay and local exponential convergence of $\|\kktConds{}(\totalvariter{}{\iter})\|$.
  }
  \label{fig:convergence}
\end{minipage}
\vspace{-2.5em}
\end{figure*}

Our experimental results are shown in \cref{fig:merging-experiments}.
These scenarios illustrate how hierarchy structure influences merging behavior.
When vehicle 3 is a follower of vehicle 1, vehicle 1 exploits its leadership position by accelerating to remain ahead after the merge.
In contrast, under a Nash relationship between vehicles 1 and 3, vehicle 1 behaves more conservatively and yields, allowing vehicle 3 to merge ahead.
Despite these differing outcomes, all equilibria in our experiments (e.g., Nash, Stackelberg, and mixed hierarchies) result in safe interactions among all vehicles, including vehicles 2 and 4.
Finally, \cref{fig:convergence} provides empirical corroboration of the convergence results in \Cref{thm:convergence-analysis}\cmt{, showing constant decrease further from the equilibrium and exponential convergence near it.}

%

Overall, this convoy example demonstrates that mixed-hierarchy games naturally encode asymmetric information access within an optimization-based framework for decision-making. 
As the hierarchy structure changes, our solver converges to distinct, interpretable behavior for this simulated example with hundreds of variables, nonquadratic objective functions, and nonlinear constraints.
\section{Conclusion}
\label{eq:discussion}


We define mixed-hierarchy games, which contain agents with both Nash and Stackelberg relationships to one another.
We present a formal definition of the relevant equilibrium concept for this type of game, and develop 
a tractable procedure for constructing the KKT conditions of mixed-hierarchy games using quasi-policy approximations to embed followers' best response policies into their leaders' optimization problems. 
We propose an inexact Newton's method for efficiently solving those approximate KKT systems. 
We provide a theoretical proof and an empirical study for the local exponential convergence of our algorithm to an approximate mixed-hierarchy equilibrium.
We evaluate our method on two complex multi-agent scenarios, and showcase its broad functionality and real-time performance using our open-source solver implementation.
Future work will explore the impact of different information structures on robot behavior\cmt{, incorporate inequality constraints via barrier methods, and relax the forest structure assumption (Assumption~\ref{ass:forest-info-structure}) to allow multi-leader mixed-hierarchies.}
\bibliography{lab_references,references}


\end{document}